\documentclass[conference,a4paper]{IEEEtran}

\usepackage[utf8]{inputenc}
\usepackage[T1]{fontenc}
\usepackage{url}
\usepackage{ifthen}
\usepackage{cite}
\usepackage{subfigure}
\usepackage{graphicx}
\usepackage[dvipsnames]{xcolor}
\usepackage{balance}
\usepackage[cmex10]{amsmath} 
\usepackage{amssymb}

\newtheorem{theorem}{Theorem}

\def\QED{~\rule[-1pt]{5pt}{5pt}\par\medskip}
\newenvironment{proof}{\noindent{\bf Proof: }}{\hspace*{\fill}\QED}

\def\hT[#1]{H_{#1}{^{(T)}}}
\def\hR[#1]{H_{#1}{^{(R)}}}
\def\hm1[#1]{H_{#1}{^{-1}}}
\def\bk[#1]{{^{(#1)}}}

\def\mbx{\mathbf{x}}
\def\mbc{\mathbf{c}}
\def\mby{\mathbf{y}}
\def\mbw{\mathbf{w}}
\def\mbz{\mathbf{z}}
\def\mbu{\mathbf{u}}
\def\mbv{\mathbf{v}}
\def\mbf{\mathbf{f}}

\def\mbt{\mathbf{t}}

\def\mbh{\mathbf{h}}
\def\mbz{\mathbf{z}}

\newcommand{\mbR}{\mathbb{R}}

\newcommand{\mcN}{\mathcal{N}}

\newcommand{\norm}[1]{\left\lVert #1 \right\rVert}
\renewcommand{\baselinestretch}{1}

\begin{document}
\title{An Information-Theoretic Explanation for the Adversarial Fragility of AI Classifiers}

\author{%
  \IEEEauthorblockN{Hui Xie, Jirong Yi, Weiyu Xu, and Raghu Mudumbai}
  \IEEEauthorblockA{Department of Electrical and Computer Engineering, University of Iowa}
}

\maketitle


\begin{abstract}
We present a simple hypothesis about a compression property of artificial intelligence (AI) classifiers and present theoretical arguments to show that this hypothesis
successfully accounts for the observed fragility of AI classifiers to small adversarial perturbations. We also propose
a new method for detecting when small input perturbations cause classifier errors, and show theoretical guarantees
for the performance of this detection method. We present experimental results with a voice recognition system
to demonstrate this method. The ideas in this paper are motivated by a simple analogy between AI classifiers and
the standard Shannon model of a communication system. \footnote{The first two authors contributed equally.}
\end{abstract}

\section{Introduction}
Recent advances in machine learning have led to the invention of complex classification systems that are very successful in
detecting features in datasets such as images, hand-written texts, or audios. However, recent works have also discovered 
what appears to be a universal property of AI classifiers: vulnerability to small adversarial perturbations. Specifically, we know 
that it is possible to design ``adversarial attacks'' that manipulate the output of AI classifiers arbitrarily by making small 
carefully-chosen modifications to the input. Many such successful attacks only require imperceptibly small perturbations of 
the inputs, which makes these attacks almost undetectable. Thus AI classifiers exhibit two seemingly contradictory properties: 
(a) high classification accuracy even in very noisy conditions, and (b) high sensitivity to very small adversarial perturbations. 
In this paper, we will use the term ``adversarial fragility'' to refer to this property (b).

The importance of the adversarial fragility problem is widely recognized in the AI community and there now exists a vast and
growing literature studying this property, see e.g. \cite{akhtar_threat_2018} for a comprehensive survey. This work, however, has
not yet resulted in a consensus on two important questions: (a) a theoretical explanation for adversarial fragility, and (b)
a general and systematic defense against adversarial attacks. Instead, we currently have multiple competing theoretical
explanations, multiple defense strategies based on both theoretical and heuristic ideas and many methods for generating 
adversarial examples for AI classifiers. Theoretical hypotheses from the literature include (a) quasi-linearity/smoothness 
of the decision function in AI classifiers \cite{goodfellow_explaining_2014}, (b) high curvature of the decision boundary \cite{fawzi_robustness_2016} and (c) closeness of the classification boundary to the data sub-manifold \cite{tanay_boundary_2016}. Defenses against adversarial attacks have also evolved from early methods using gradient masking \cite{papernot_practical_2017}, 
to more sophisticated recent methods such as adversarial training where an AI system is specifically subjected to adversarial 
attacks as part of its training process \cite{tramer_ensemble_2017}, and defensive distillation \cite{papernot_distillation_2016}. 
These new defenses in turn motivate the development of more sophisticated attacks \cite{carlini_defensive_2016} in an 
ongoing arms race.

In this paper, we show that property ``adversarial fragility'' is an unavoidable consequence of a simple ``compression'' 
hypothesis about AI classifiers. This hypothesis is illustrated in Fig. \ref{Fig2}: we assume that the output of AI classifiers 
is a function of a highly compressed version of the input. More precisely, we assume that the output of AI classifiers is a 
function of an intermediate set of variables of much smaller dimension than the input. The intuition behind this hypothesis 
is as follows. AI classifiers typically take high-dimensional inputs e.g. image pixels, audio samples, and produce a discrete 
{\it label} as output. The input signals (a) contain a great deal of redundancy, and (b) depend on a large number of irrelevant 
variables that are unrelated to the output labels. Efficient classifiers, therefore, must remove a large amount of redundant 
and/or irrelevant information from the inputs before making a classification decision. Indeed, a classifier that {\it generalizes} 
well, must, by definition, be insensitive to as many  non-essential input features as possible. We show in this paper that 
adversarial fragility is an immediate and necessary consequence of this ``compression'' property.

Certain types of AI systems can be shown to satisfy the compression property simply as a consequence of their structure.
For instance, AI classifiers for the MNIST dataset \cite{lecun_gradient-based_1998} typically feature a final layer 
in the neural network architecture that consists of softmax over a ${10 \times 1}$ real-numbered vector corresponding to the $10$ different 
label values; this amounts to a substantial dimension reduction from the $28 \times 28 = 784$ dimensional pixel vector 
at the inputs. More generally, there is some empirical evidence showing that AI classifiers actively compress their inputs during 
their training process \cite{shwartz-ziv_opening_2017}.

Our proposed explanation of adversarial fragility also immediately leads to an obvious and very powerful 
defense: if we enhance a classifier with a generative model that at least partially ``decompresses'' the classifier's output,
and compare it with the raw input signal, it becomes easy to check when adversarial attacks produce classifier outputs that are
inconsistent with their inputs. While we present some 
simple experimental results to validate our theory, our focus here is on the theoretical ideas; the important and challenging problem 
of designing good generative models to implement the proposed defense 
for general AI classification systems is deferred to future work. Interestingly, while our theory is novel, other researchers 
have recently developed defenses for AI classifiers against adversarial attacks that are consistent with our proposed approach
\cite{kundu_bihmp-gan:_2018, frosst_darccc:_2018}. 



\section{Problem Statement}
\begin{figure}[htb]
	\begin{center}
		\includegraphics[scale=0.24]{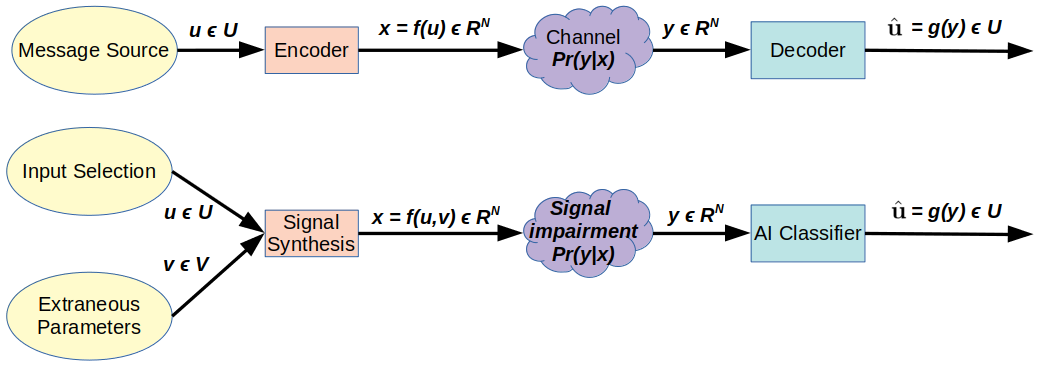}
	\end{center}
	\caption{Top: standard abstract model of a communication system; Bottom: abstract model of an AI classifier system.}
	\label{Fig1}
\end{figure}

An AI classifier can be defined as a system that takes a high-dimensional vector as input and maps it to a discrete set of labels.
As an example, a voice-recognition AI takes as input a {{time series}} containing the samples of an audio signal and outputs a string 
representing a sentence in English (or other spoken language). More concretely, consider Fig. \ref{Fig1} which explores a simple 
analogy between an AI classification system and a digital communication system.

The purpose of the AI system in Fig. \ref{Fig1} is to estimate the state of the world $\mathbf{u} \in \mathcal{X}$ where the set 
of all possible world states $\mathcal{X}$ is assumed to be finite and are enumerated as $\mbu_1,~\mbu_2 \dots,~ \mbu_{N_u}$, 
where $N_u$ is the size of $\mathcal{X}$. The input $\mathbf{y} \in \mathbb{R}^{N}$ to the AI classifier is a noisy version of signals $\mathbf{x} 
\in \mathbb{R}^N$, and $\mbx$ depend on $\mathbf{u}$ and on a number of {\it extraneous parameters} $\mathbf{v} \in \mathcal{V}$.
Note that the state $\mbu_i$ is uniquely determined by its index or ``label'' $i$. The output of the AI classifier is a state estimate 
$\hat{\mathbf{u}} \in \mathcal{X}$, or equivalently, its label.

The AI classifier in Fig. \ref{Fig1} is clearly analogous to a communication decoder: it looks at a set of noisy observations and 
attempts to decide which out of a set of possible input signals $\mathbf{x}$ was originally ``transmitted'' over the ``channel'', 
which in the AI system models all signal impairments such as distortion, random noise and hostile attackers. 

The ``Signal Synthesis'' block in the AI system maps input features into an observable signal $\mbx$. In the abstract model of Fig. 
\ref{Fig1}, the synthesis function $\mbf(\cdot)$ is deterministic with all random effects being absorbed into the ``channel'' without loss 
of generality. Note that while the encoder in the communication system is under the control of its designers, the 
signal synthesis in an AI system is determined by physical laws and is not in our control. However, the most important difference 
between communication and AI systems is the presence of the nuisance parameters $\mbv$. For instance, in a voice recognition 
system, the input features consist of the text being spoken ($\mbu$) and also a very large number of other characteristics ($\mbv$) 
of the speaker's voice such as pitch, accent, dialect, loudness, emotion etc. which together determine the mapping from a text to 
an audio signal. Thus there are a very large number of different ``codewords'' $\mbc_1 = \mbf(\mbu_1,\mbv_1),~\mbc_2 = 
\mbf(\mbu_1,\mbv_2), \dots$ that encode the same label $\mbu_1$. Let us define the ``codeword set'' for label ${{i}},~i=1 \dots N_u$:
\begin{align}
\mathcal{X}_i &\doteq \{ \mbc \in \mathbb{R}^N : \exists \mbv,~\mbc = \mbf(\mbu_i, \mbv) \} \label{eq:sdef1}
\end{align}
We assume that the codeword sets $\mathcal{X}_i$ satisfy:
\begin{align}
\min_{\forall i,j,~i \neq j} \min_{\mbc_i \in \mathcal{X}_i,~\mbc_j \in \mathcal{X}_j } \norm{\mbc_i - \mbc_j} \geq 2r_0 \label{eq:sep}
\end{align}
for some $r_0>0$, where $\|\cdot\|$ represents $\ell_{2}$ norm. In other words, all valid codewords corresponding to different labels 
$i \neq j$ are separated by at least a distance $2r_0$. In the voice recognition example, under this assumption audio signals corresponding to 
two different sentences must sound different. This guarantees the existence of the ideal classifier defined as the function 
$q^*(\mbx): \mathbb{R}^N \rightarrow \mathcal{X}$ that satisfies $q^* \left( \mbf(\mbu_i, \mbv) \right) = \mbu_i,~\forall i,~\mbv 
\in \mathcal{V}$. By definition, the ideal classifier maps any valid input signal to the correct label in the absence of noise. 

\begin{figure}
\begin{center}
	\includegraphics[scale=0.43]{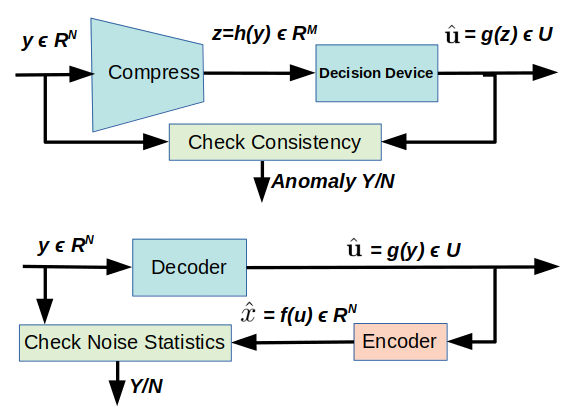}
	\end{center}
	\caption{AI classifier using a information compression process and its analogy with a communication decoder}
	\label{Fig2}
\end{figure}

Fig. \ref{Fig2} shows an abstract model of a classifier that is constrained to make final classification decisions based on only  a compressed version $\mbz$ of $\mby$. Specifically, we assume that there exists a compression function $\mbh: \mathbb{R}^N \rightarrow \mathbb{R}^M$, 
where $M \ll N$ such that the classifier output $q(\mby): \mathbb{R}^N \rightarrow \mathcal{X}$ can be written as $q(\mby) = g(\mbh(\mby))$, where $g: \mathbb{R}^M \rightarrow \mathcal{X}$ is a decision function. We define the ``compressed codeword 
sets'' as $\mathcal{Z}_i \doteq \{ \mbz \in \mathbb{R}^M: \exists \mbu \in \mathcal{U}_i, \mbv \in \mathcal{V},~\mbh(\mbf(\mbu,\mbv))=\mbz \}$. We will assume that the sets $\mathcal{Z}_i$ are disjoint so that 
the compression map $\mbh(\mby)$ preserves information in $\mby$ about the label $i$.

We will show that a classifier constrained to use only $h(\mby)$ for decoding, even if designed optimally, can retain its robustness to random 
noise $\mbw$, but is necessarily vulnerable to adversarial attacks that are significantly smaller in magnitude. By contrast, uncompressed 
classifiers can be robust to both random and worse-case noise. In other words, we show that adversarial fragility can be explained as an artifact of compression or dimension reduction in decoders.

Our method for detecting adversarial attacks is based on the idea of at least partially ``decompressing'' the output of the classifier and 
checking it for consistency against the raw observations $\mby$. Specifically, suppose the classifier outputs label $j$ for input signal $\mby$.
Define $\mbc_j(\mby)$ as:
\begin{align}
\mbc_j(\mby) &\doteq \arg \min_{\mbc \in \mathcal{X}_j} \norm{\mby - \mbc} \label{eq:cidef}
\end{align}
If we observe that $\norm{\mby - \mbc_j(\mby)}$ is abnormally large, this means that the observed signal $\mby$ is far from any valid 
codeword $\mbf(\mbu_j,\mbv)$ with label $j$ and we conclude that label $j$ is inconsistent with observations $\mby$. This, however, requires 
a feasible method for calculating $\mbc_j(\mby)$ for a label $j$ and signal $\mby$. This is basically a {\it denoising} operation that outputs a 
noise-free codeword $\mbc_j(\mby)$ given a label $j$ and noisy observation $\mby$. Generative models are capable of performing such a 
{\it denoising} operation. We do not discuss the design of such models in this paper; instead we will show that under mild assumptions on 
the encoding function $\mbf(\cdot)$, we can provide theoretical guarantees on a detector {\it assuming a well-functioning generative model}.

\section{Theoretical Analysis}
In this section, we consider a signal $\mbx \in \mathbb{R}^{N} $,  which can be a noisy version of a codeword.  Without loss of generality, we assume that an ideal classifier will classify $\mbx$ to label $1$, and assume that the closest codeword to $\mbx$ is $\mbc_{1}=\mbf(\mbu_1,\mbv_1)$ for $\mbu_{1}$ and  a certain $\mbv_{1}$.
For any $i\neq 1$, we also define $\mbc_{i}$ as the codeword with label $i$ that is closest to $\mbx$: $\mbc_i \doteq \arg \min_{\mbc \in \mathcal{X}_i} \norm{\mbx-\mbc}$.
We define the sets $\mathcal{S}_1$ and $\mathcal{S}_i$ as the spheres of size $r$ around $\mbc_{1}$ and $\mbc_i$ respectively, namely $\mathcal{S}_1 \doteq \{ \mathbf{b} \in \mathbb{R}^N: \norm{\mathbf{b}-\mbc_{1}} < r \}$ and
$\mathcal{S}_i \doteq \{ \mathbf{b} \in \mathbb{R}^N: \norm{\mathbf{b}-\mbc_i} < r \}$.  We assume that $\mbx \in \mathcal{S}_{1}$.  For simplicity of analysis,  we assume that, for a vector $\mby \in \mathbb{R}^{N}$, the classifier $q(\mathbf{y})$ outputs label $i$ if and only if $h(\mathbf{y})=h(\mathbf{b})$ for  a certain  $\mathbf{b}\in \mathcal{S}_{i}$. 

We consider the problem of finding the smallest targeted
perturbation $\mbw$ in magnitude which fools the decoder $q(\mbx+\mbw)$ into outputting label $i \neq 1$. Formally, for any $\mbx \in \mathbb{R}^N$, we define the
minimum perturbation size $d_i(\mbx)$ needed for target label $i$ as:
\begin{align}
d_i(\mbx) \doteq \min_{\mbw\in\mathbb{R}^N, \mbt \in \mathcal{S}_i}~\norm{\mbw},~\text{s.t.}~\mbh(\mbx+\mbw)=\mbh(\mbt).
\label{eq:optimizationknownparameter1}
\end{align}

Let us define a quantity $d(\mbx, \mbt)$, which we term as ``effective distance between $\mbx$ and $\mbt$ with respect to function
$\mbh(\cdot)$'' as $d(\mbx, \mbt) = \min_{\mbw \in \mathbb{R}^{N},~\mbh(\mbx+\mbw)=\mbh(\mbt)}  \|\mbw\|$. Then for any vector
$\mbx \in \mathbb{R}^N$, we can use (\ref{eq:optimizationknownparameter1}) to upper bound the smallest required perturbation
size $d_i(\mbx) \leq \min_{\mbt \in \mathcal{S}_{i}} d(\mbx, \mbt)$.

For an $\epsilon>0$ and $l>0$, we say a classifier has $(\epsilon, l)$-robustness at signal $\mbx$, if $\mathbb{P}( g(\mbh(\mbx+\mbw))=g(\mbh(\mbx)))  \geq 1-\epsilon$, where $\mbw \in \mathbb{R}^{N}$ is randomly sampled uniformly on a sphere\footnote{Defined
for some given norm, which we will take to be $\ell_2$ norm throughout this paper.} of radius $l$, and $\mathbb{P}$ means probability. In the following, we will show that for a small
$\epsilon$, compressed classifiers can still have $(\epsilon, l)$-robustness for $l \gg d_i(\mbx)$, namely the classifier can tolerate large random
perturbations while being vulnerable to much smaller adversarial attacks.

\subsection{Classifiers with Linear Compression Functions}

We first consider the special case where the compression function $\mbh(\cdot)$ is linear, namely $\mbh(\mby)=A \mby$ with $A \in
\mathbb{R}^{M \times N},~M \ll N$. While this may not be a reasonable model for practical AI systems, analysis of linear compression
functions will yield analytical insights that generalize to nonlinear $\mbh(\cdot)$ as we show later.

\begin{theorem}
Let $\mby \in \mathbb{R}^{N}$ be the input to a classifier, which makes decisions based on the compression function $\mbz= \mbh(\mby)= A \mby$, where the elements of $A \in \mathbb{R}^{M \times N}$  ($M\ll N$)
are i.i.d. following the standard Gaussian distribution $\mathcal{N}(0,1)$. Let $B_i=\{\mbz~:~ \mbz=A\mathbf{b}, \mathbf{b}\in \mathcal{S}_i \}$ be the
compressed image of $\mathcal{S}_i$. Then the following statements hold for arbitrary $\epsilon>0$, $i\neq 1$, and a big enough $M$.\\
1) With high probability (over the distribution of $A$), an attacker can design a targeted adversarial attack $\mbw$ with
$\|\mbw\|_2 \leq \sqrt{1+\epsilon}\sqrt{\frac{M}{N}} \|\mbc_{i}-\mbx\|_{2} -r$ such that the classifier is fooled into classifying the signal $\mby=\mbx+\mbw$ into label $i$. Moreover, with high probability (over the distribution of $A$), an attacker can design an (untargeted) adversarial perturbation $\mbw$ with  $\|\mbw\|_2 \leq r- \sqrt{1-\epsilon} \sqrt{\frac{M}{N}} \|\mbx-\mbc_1\|$ such that the classifier will not classify $\mby=\mbx+\mbw$ into label $1$.\\
2) Suppose that $\mbw$ is randomly uniformly sampled from a sphere of radius $l$ in $\mathbb{R}^N$. With high probability (over the distribution of $A$ and $\mbw$), if $l<\sqrt{\frac{1-\epsilon}{1+\epsilon}}  \|\mbc_{i}-\mbx\|_{2}-\frac{r}{\sqrt{1+\epsilon} \sqrt{\frac{M}{N}}}$, the classifier will not classify $\mby=\mbx+\mbw$ into label $i$.
Moreover, with high probability (over the distribution of $A$ and $\mbw$),  if $l<(1-\epsilon)\sqrt{\frac{N}{M}} \sqrt{r^2-\frac{M}{N} \|\mbx-\mbc_1\|^2 }$, the classifier still classifies the $\mby=\mbx+\mbw$ into label $1$ correctly.\\
3) Let $\mbw$ represent a successful adversarial perturbation i.e. the classifier outputs target label $i\neq 1$ for the input $\mby=\mbx+\mbw$. Then as long as $\|\mbw\|_{2} < \min_{\mbc_i \in \mathcal{X}_i}\|\mbc_{i} -\mbx\|-r$, our adversarial detection approach will be able to detect the attack.
\end{theorem}

\begin{proof}
1) We first look at the targeted attack case. For linear decision statistics, $d(\mbx, \mbt) = \min_{\mbw \in \mathbb{R}^{N}, A(\mbx+\mbw)=A(\mbt)}  \|\mbw\| $.
Solving {{this optimization problem}},
we know the optimal $\mbw$ is given by
$\mbw=A^{\dagger}  A(\mbt-\mbx),$
where $A^{\dagger}$ is the Moore-Penrose inverse of $A$.
We can see that $\mbw$ is nothing but the projection of $(\mbt-\mbx)$ onto the row space of $A$.  We denote the projection matrix as $P=A^{\dagger}  A$.
Then the smallest magnitude of an effective adversarial perturbation is upper bounded by
$$ \min_{\mbt \in \mathcal{S}_{i}} d(\mbx, \mbt)
=   \min_{\mbt \in \mathcal{S}_{i}} \|   A^{\dagger}  A(\mbt-\mbx) \|
= \min_{\mbt \in \mathcal{S}_{i}} \|   P(\mbt-\mbx) \|     .$$
For $\mbt \in \mathcal{S}_{i}$, we have
$\|   P(\mbt-\mbx) \|=  \|   P (\mbc_{i}-\mbx)+ P (\mbt-\mbc_{i}) \| \geq   \|   P (\mbc_{i}-\mbx)\|-    \|P (\mbt-\mbc_{i}) \|$. One can show that, when  $\mathcal{S}_{i}=\{\mbt~|~\|\mbt-\mbc_{i}\|\leq r\}$, we can always achieve the equality, namely $ \min_{\mbt \in \mathcal{S}_{i}} \|   P(\mbt-\mbx) \| =  \|   P (\mbc_{i}-\mbx)\|-r$.

Now we evaluate $\|P (\mbc_{i}-\mbx)\|$. Suppose that $A$'s elements are i.i.d., and follow the standard zero-mean Gaussian distribution $\mathcal{N} (0,1)$, then the random projection $P$ is uniformly sampled from the Grassmannian $Gr(M, \mathbb{R}^{N})$.  We can see that the distribution of  $\|   P (\mbc_{i}-\mbx)\|$ is the same as the distribution of the magnitude of the first $M$ elements of $\|   (\mbc_{i}-\mbx)\| \mathbf{o}/\|\mathbf{o}\| $, where $\mathbf{o} \in \mathbb{R}^N$ is a vector with its elements being i.i.d.  following the standard Gaussian distribution  $\mathcal{N}(0,1)$.  From the concentration of measure, for any positive $\epsilon<1$,
\begin{align*}
& \mathbb{P} \left(\|P(\mbc_{i}-\mbx)\|
\leq \sqrt{1-\epsilon} \|\mbc_{i}-\mbx\| \sqrt{\frac{M}{N}}\right) \leq e^{- \frac{M\epsilon^2}{4}},\\
& \mathbb{P} \left(\|P(\mbc_{i}-\mbx)\|
\geq \sqrt{1+\epsilon} \|\mbc_{i}-\mbx\| \sqrt{\frac{M}{N}}\right) \leq e^{- \frac{M\epsilon^2}{12}}  .
\end{align*}
Then when $M$ is big enough, $ \min_{\mbt \in \mathcal{S}_{i}} \|   P(\mbt-\mbx) \| \leq  \sqrt{1+\epsilon}\sqrt{\frac{M}{N}} \|\mbc_{i}-\mbx\| -r $ with high probability, for arbitrary $\epsilon>0$.

Now let us look at what perturbation $\mbw$ we need such that $A(\mbx+\mbw)$ is not in $B_1$. One can show that $A(\mbx+\mbw)$ is outside $B_1$ if and only if, $\|P(\mbx-\mbc_1+\mbw)\| >r$.  Then by the triangular inequality, the attacker can take an attack $\mbw$ with $\|\mbw\| > r- \|  P(\mbx-\mbc_1) \|$, which is no bigger than
$r- \sqrt{1-\epsilon} \sqrt{\frac{M}{N}} \|\mbx-\mbc_1\|_2$ with high probability, for arbitrary $\epsilon>0$ and big enough $M$.

2) If and only if $\mbh(\mbx +\mbw )\neq \mbh(\mbt)$, $\forall  \mbt \in \mathcal{S}_{i} $,  $\mbw$ will not fool the classifier into label $i$.
If $\mbh(\mby)=A \mby$, ``$\mbh(\mbx +\mbw )\neq \mbh(\mbt)$, $\forall  \mbt \in \mathcal{S}_{i} $'' is equivalent to
``$ \|A(\mbx+\mbw-\mbt)\| \neq 0  $, $\forall  \mbt \in \mathcal{S}_{i} $'', which is in turn equivalent to
``$ \|P(\mbx+\mbw-\mbt)\| \neq 0  $, $\forall  \mbt \in \mathcal{S}_{i} $'', where $P$ is the projection onto the row space of $A$. Assuming that $\mbw$ is uniformly randomly sampled from a sphere in $\mathbb{R}^{N}$ of radius $l<\sqrt{\frac{1-\epsilon}{1+\epsilon}}  \|\mbc_{i}-\mbx\|-\frac{r}{\sqrt{1+\epsilon} \sqrt{\frac{M}{N}}}$, then
\begin{align*}
\|P(\mbx+\mbw-\mbt )\|
& =\|P(\mathbf{c}_{i}-\mbx)+P(\mbt-\mbc_{i}) -P\mbw \|\\
&\geq \| P(\mathbf{c}_{i}-\mbx)   \|-\| P(\mbt-\mathbf{c}_{i})  \|-\|P\mbw\|.
\end{align*}

From the concentration inequality,
$\mathbb{P} \left(\|P\mbw\| \geq \sqrt{1+\epsilon} \|\mbw\| \sqrt{\frac{M}{N}}\right)\leq e^{- \frac{M(\epsilon^2/2-\epsilon^3/3 )   }{2}}$.
Thus if $M$ is big enough, with high probability,
$\|P(\mbx+\mbw-\mbt )\| \geq \sqrt{1-\epsilon} \|\mbc_{i}-\mbx\|_{2} \sqrt{\frac{M}{N}} -r
 -\sqrt{1+\epsilon} \|\mbw\| \sqrt{\frac{M}{N}}$. If $\|\mbw\| = l$, $\|P(\mbx+\mbw-\mbt )\|>0$.


Now let us look at what magnitude we need for a random perturbation $\mbw$ such that $A(\mbx+\mbw)$ is in $B_1$ with high probability. We know $A(\mbx+\mbw)$ is in $B_1$ if and only if, $\|P(\mbx-\mbc_1+\mbw)\| <r$.  Through a large deviation analysis, one can show that, for any $\delta>0$ and big enough $M$,  $\|P(\mbx-\mbc_1+\mbw)\|$ is smaller than $(1+\delta)\sqrt{\frac{M}{N} \|\mbx-\mbc_1\|^2+ \frac{M}{N} l^2}$  and bigger than $(1-\delta)\sqrt{\frac{M}{N} \|\mbx-\mbc_i\|^2+ \frac{M}{N} l^2}$ with high probability. Thus, for an arbitrary $\epsilon>0$, if $l<(1-\epsilon)\sqrt{\frac{N}{M}} \sqrt{r^2-\frac{M}{N} \|\mbx-\mbc_1\|^2 }$, $\|P(\mbx-\mbc_1+\mbw)\|<r$ with high probability, implying the AI classifier still classifies the $\mby=\mbx+\mbw$ into Class $1$ correctly.


3) Suppose that an AI classifier classifies the input signal $\mby=\mbx+\mbw$ into label $i$.  We propose to check whether $\mby$ belongs to $\mathcal{S}_{i}$. In our model,  the signal  $\mby$ belongs to $\mathcal{S}_{i}$ only if $ \min_{\mbc_{i} \in \mathcal{X}_{i}   }\|\mby-\mbc_{i}\|\leq r$. Let us take any codeword $\mbc_{i} \in \mathcal{X}_{i}$.   We show that when $\|\mbw\| < \|\mbc_{i} -\mbx\|-r$, we can always detect the adversarial attack if the AI classifier misclassifies $\mby$ to that codeword corresponding to label $i$. In fact, $\|\mby-\mbc_{i}\|
=\|(\mbx  +\mbw)-\mbc_{i}  \|_{2}$, which is no smaller than $\|  \mbc_{i}- \mbx  \| -\|\mbw\|\geq \|  \mbc_{i}- \mbx  \|-(\|\mbc_{i} -\mbx\|-r)
>r$.

We note that $ \|\mbw\| \leq  \min_{\mbc_{i} \in \mathcal{X}_i   }\|\mbc_{i}-\mbx\|-r$ means $\|\mbw\|_{2} < \|\mbc_{i} -\mbx\|-r$ for every codeword $\mbc_{i}$, thus implying that the adversary attack detection technique can detect that $\mby$ is at more than distance $r$ from every codeword from $\mathcal{X}_{i}$.

\end{proof}

\subsection{Nonlinear Decision Statistics in AI Classifiers}
In this subsection, we show that an AI classifier using nonlinear compressed decision statistics $\mbh(\mby) \in \mathbb{R}^{M}$is significantly more vulnerable to adversarial attacks than to random perturbations. We will quantify the gap between how much a random perturbation and a well-designed adversarial attack affect  $\mbh(\mby)$, the full proof of which is in \cite{full_version_paper_2019}.

\begin{theorem}\label{Thm:Nonlinear}
Let us assume that the nonlinear function $\mbh(\mbx):\mbR^N\to \mbR^M$ is differentiable at $\mbx$.  For $\epsilon>0$, we define $\alpha (\epsilon)=\max_{\|\mbw\| \leq  \epsilon} (\|\mbh(\mbx+\mbw)-  \mbh(\mbx)\| )$, and $\beta (\mathbf{o}, \epsilon)=  \|\mbh(\mbx+{{\epsilon}}\mathbf{o})-  \mbh(\mbx)\|$, where $\mathbf{o}$ is  uniformly randomly sampled from a unit sphere.  Then $\lim_{\epsilon \rightarrow 0}   \frac{ \alpha (\epsilon) }{ E_{\mathbf{o}} \{ \beta (\mathbf{o}, \epsilon) \} }  \geq   \sqrt{\frac{N}{M}} $, where $E_{\mathbf{o}} $ means expectation over the distribution of $\mathbf{o}$. If we  assume that the entries of the Jacobian matrix $\nabla \mbh(\mbx)\in\mbR^{M\times N}$ are i.i.d. distributed following the standard Gaussian distribution $\mcN(0,1)$, then, when $N$ is big enough, with high probability, $\lim_{\epsilon \rightarrow 0}   \frac{ \alpha (\epsilon) }{ E_{\mathbf{o}} \{ \beta (\mathbf{o}, \epsilon) \} }  \geq  (1-\delta) \sqrt{\frac{N+M}{M}} $ for any $\delta>0$, .
\end{theorem}

\section{Experimental Results}
\begin{figure}
\begin{center}
	\includegraphics[scale=0.2]{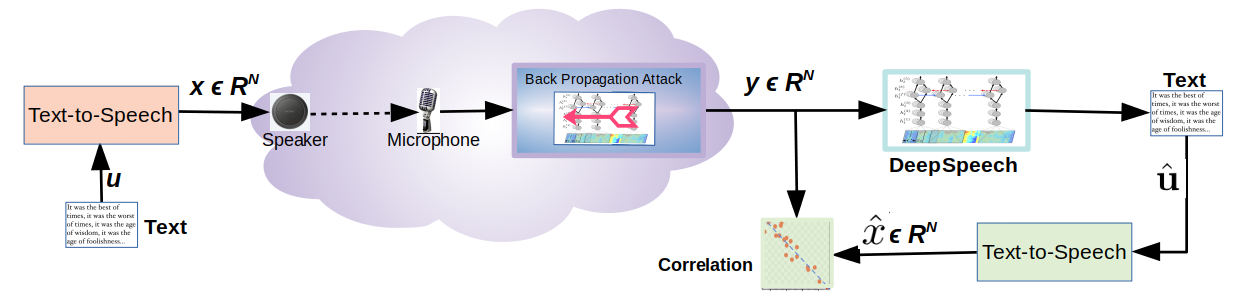}
	\end{center}
	\caption{System Modules of Using Correlation Coefficients to Detect Adversarial Attacks}
	\label{Fig3}
\end{figure}

We performed a series of experiments\footnote{https://github.com/Hui-Xie/AdversarialDefense} to test and illustrate our proposed defense for a popular voice recognition AI classifier DeepSpeech\footnote{https://github.com/mozilla/DeepSpeech}. The experimental setup is illustrated in Fig. \ref{Fig3}; a visual comparison with the abstract model in Fig.\ref{Fig1} shows how the various functional blocks are implemented in our experiment.


The experiment consisted of choosing sentences randomly from the classic 19-th century novel ``A Tale of Two Cities.''  A Linux text-to-speech (T2S) software, Pico2wave, converted a chosen sentence e.g. $\mbu_1$ into a female voice wave file. The use of a T2S system for generating the source audio signal (instead of human-spoken audio) effectively allows us to hold the all ``irrelevant'' variables $\mbv$ constant, and thus renders the signal synthesis block in Fig. \ref{Fig1} as a deterministic function of just the input label $\mbu_1$. 


\begin{figure}
	\includegraphics[scale=0.32]{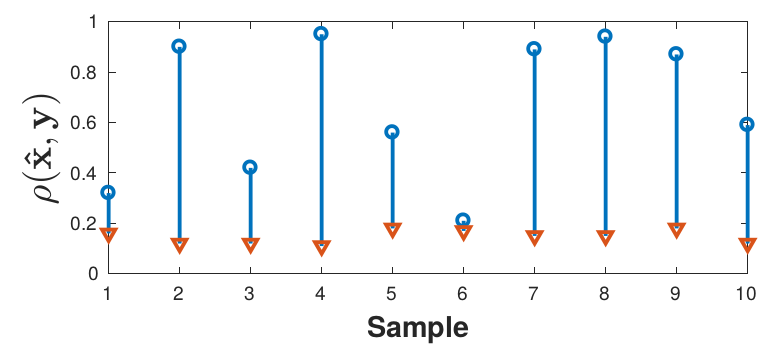}
	\caption{The change of cross correlation coefficients $\rho(\hat{x},y)$.  Blue circles indicate $\rho$ between input signals $y_1$ without adversarial attack  and their corresponding reconstructed signals $\hat{x}$ from decoded labels,  and red triangles indicate $\rho$ between input signals $y_2$ with adversarial attacks and its coresponding reconstructed signals $\hat{x}$. Low ``blue circles'' mean DeepSpeech runs into recognition failure in several error characters, even if no adversarial attacks are present. }
	\label{Fig4}
\end{figure}

Let $x[n]$ denote the samples of this source audio signal. This audio signal is played over a PC speaker and recorded by a USB microphone on another PC. Let $y_1[n]$ denote the samples of this recorded wave file. The audio playback and recording was performed in a quiet room with no audible echoes or distortions, so this ``channel'' can be approximately modeled as a simple AWGN channel: $y_1[n]= \alpha x[n] +  w_1[n]$, where $\alpha$ is a scalar representing audio signal attenuation and $w_1[n]$ is random background noise. In our experiment, the $\mathrm{SNR} \doteq \frac{\alpha^{2}\norm{x}^2} {\norm{w_1}^2}$ was approximately $28$ dB.

We input $y_1[n]$ into a voice recognition system, specifically, the Mozilla implementation DeepSpeech V0.1.1 based on TensorFlow. The 10 detailed sentences are demonstrated in the table of our full paper \cite{full_version_paper_2019}. We then used Nicholas Carlini's adversarial attack Python script\footnote {https://nicholas.carlini.com/code} with Deep Speech (V0.1.1) through gradient
back-propagation to generate a targeted adversarial audio signal $y_2[n]=y_1[n]+w_2[n]$ where $w_2[n]$ is a small adversarial perturbation that causes the DeepSpeech voice recognition system to predict a completely different
sentence $\mbu_2$. Thus, we have a ``clean'' audio signal $y_1[n]$, and a ``targeted corrupted'' adversarial audio signal $y_2[n]$ that upon playback is effectively indistinguishable from $y_1[n]$, but successfully fools DeepSpeech into outputting a different target sentence. In our experiment, the power of $y_2[n]$ over  the adversarial perturbation $w_2$ was approximately $35$ dB.

We then implemented a version of our proposed defense to detect whether the output of the DeepSpeech is wrong, whether due to noises or adversarial attacks. For this purpose, we fed the decoded text output of the DeepSpeech system into the same T2S software Pico2Wave, to generate a reconstructed female voice wave file,  denoted by $\hat{x}[n]$. We then performed a simple cross-correlation of a portion of the reconstructed signal (representing approximately $10\%$ reconstruction of the original number of samples in $x[n]$) with the input signal $y[n]$ to the DeepSpeech classifier:
$\rho_{max}(\hat{x},y) = \max_m \left| \sum_n \hat{x}[n] y[n-m] \right|$.
If $\rho_{max}$ is smaller than a threshold (0.4), we declare that the speech recognition classification is wrong. The logic behind this test is as follows. When the input signal is $y_1[n]$ i.e. the non-adversarial-perturbed signal, the DeepSpeech successfully outputs the correct label $\hat{\mbu} \equiv \mbu_1$,which results in $\hat{x}[n] \equiv x[n]$. Since $y_1[n]$ is just a noisy version of $x[n]$, it will be highly correlated with $\hat{x}[n]$. On the other hand, for the adversarial-perturbed input $y_2[n]$, the reconstructed signal $\hat{x}[n]$ is completely  different from $x[n]$ and therefore can be expected to be practically uncorrelated with $y_2[n]$.

Fig.\ref{Fig4} shows the cross-correlation $\rho_{max}(\hat{x},y)$ for $10$ sets of recorded signals (a) with and (b) without adversarial perturbations
(red triangles and blue circles respectively in Fig. \ref{Fig4}) . The adversarial perturbations all successfully fool the DeepSpeech AI to output the target text $\mbu_2=$``he travels the fastest who travels alone''. We see that the observed correlations for the adversarial signals are always very small, and are therefore successfully detected by our correlation test. Interestingly, some of the non-adversarial signals yield low correlations as well, but {{this is because the DeepSpeech cannot decode perfectly}} even when there are no adversarial attacks present. 






\bibliographystyle{IEEEtran}
\bibliography{Ref_adversarial}
\end{document}